\documentclass[sigconf,authorversion,nonacm]{acmart}

\usepackage{algpseudocode}
\usepackage[linesnumbered,ruled,vlined]{algorithm2e}
\usepackage{tabularx}
\usepackage{amsmath}
\usepackage{multirow}
\usepackage{graphicx}
\usepackage{subcaption}
\usepackage{enumitem}

\usepackage{amsmath, amsthm, amsfonts}
\newtheorem{theorem}{Theorem}

\def\ppdl{PPDL\xspace}
\def\ppdlbase{\ppdl-NV\xspace}
\def\ppdlerror{\ppdl-LWE\xspace}
\def\ppdlpairs{\ppdl-PW\xspace}

\usepackage{xcolor}

\AtBeginDocument{%
  \providecommand\BibTeX{{%
    \normalfont B\kern-0.5em{\scshape i\kern-0.25em b}\kern-0.8em\TeX}}}

\begin{document}

\title{Privacy-Preserving, Dropout-Resilient Aggregation in Decentralized Learning}

% Efficient Dropout-resilient Aggregation for Privacy-preserving Machine Learning

\author{Ali Reza Ghavamipour, Benjamin Zi Hao Zhao, Fatih Turkmen}


\begin{abstract}
Decentralized learning (DL) offers a novel paradigm in machine learning by distributing training across clients without central aggregation, enhancing scalability and efficiency. However, DL's peer-to-peer model raises challenges in protecting against inference attacks and privacy leaks. By forgoing central bottlenecks, DL demands privacy-preserving aggregation methods to protect data from 'honest but curious' clients and adversaries, maintaining network-wide privacy. Privacy-preserving DL faces the additional hurdle of client dropout, clients not submitting updates due to connectivity problems or unavailability, further complicating aggregation.

This work proposes three secret sharing-based dropout resilience approaches for privacy-preserving DL. Our study evaluates the efficiency, performance, and accuracy of these protocols through experiments on datasets such as MNIST, Fashion-MNIST, SVHN, and CIFAR-10. We compare our protocols with traditional secret-sharing solutions across scenarios, including those with up to 1000 clients. Evaluations show that our protocols significantly outperform conventional methods, especially in scenarios with up to 30\% of clients dropout and model sizes of up to $10^6$ parameters. Our approaches demonstrate markedly high efficiency with larger models, higher dropout rates, and extensive client networks, highlighting their effectiveness in enhancing decentralized learning systems' privacy and dropout robustness.

\end{abstract}

\maketitle

\section{Introduction}

Machine learning (ML) has become a pivotal component in numerous applications, including pattern recognition, medical diagnosis, and credit risk assessment. The effectiveness of a machine learning model is heavily dependent on the availability of large volumes of data. The conventional method of training a machine learning model involves collecting a dataset on a central server and conducting the training process on this server. However, centralization poses difficulties when the data is distributed over numerous devices or among a workforce that are located in different geographical regions. This issue is compounded by growing concerns over data privacy. To address these challenges, researchers have suggested alternative approaches to centralized learning under the umbrella term Collaborative Machine Learning (CML). CML enables the training of machine learning models directly on local devices to enhance privacy (by maintaining data locality) and offload compute demand (away from the central server)~\cite{li2020federated}.
% Collaborative Machine Learning (CML) represents a major leap forward in enhancing privacy and security in the development and refinement of ML models. 
% This approach is especially crucial in situations where centralizing data is not feasible or raises significant privacy and regulatory issues. It allows various participants to contribute to the ML models' creation and improvement while safeguarding the privacy of individual data sources.

% One of the prime examples of CML is Federated learning (FL)~\cite{mcmahan2017communication} that attracted significant attention from both academia and industry. FL distributes the machine learning training process to multiple nodes, from individual mobile devices to organizational data centers, with a designated central server coordinating the training process between nodes. Each participating node's primary task is to improve the model with its own data, before forwarding the model updates to the central server, which then combines these updates and sends them back to the clients.  While FL enables the training of ML models in settings where it was not possible previously, it places an inherent trust on the aggregation server which presents a single point of failure and leads to new challenges related to model integrity and data privacy~\cite{bonawitz2019towards}.

Federated Learning (FL)~\cite{mcmahan2017communication}, a key example of CML, has garnered notable interest across academia and industry. It decentralizes the training process across multiple clients, ranging from mobile devices to data centers, under the coordination of a central server. Clients locally update the model using their data, then share these updates with the server, which aggregates them for distribution. While FL facilitates training in novel contexts, it relies heavily on a central server, introducing risks like a single point of failure and challenges in maintaining model integrity and data privacy~\cite{bonawitz2019towards}.

Decentralized Learning (DL)~\cite{lian2017can,beltran2023decentralized,koloskova2019decentralized,hegedHus2021decentralized} emerged as a popular, scalable, and communication-efficient alternative to traditional CML algorithms. Unlike FL methods, DL operates without a central aggregator, thereby mitigating integrity and privacy risks associated with the central server in FL, and more importantly avoiding a (potentially untrusted) single point of failure. DL achieves global model training through on-device aggregation of model parameters, facilitated by peer-to-peer exchanges among clients~\cite{hegedHus2019gossip}. This process, akin to gossip-based algorithms in random networks eliminates the need for centralized aggregation. In DL, each client plays an active role in updating its model using both local data and updates from its peers, leading to improvements in model accuracy and faster convergence.

As we have established, the practice of sharing model updates instead of raw datasets mitigates privacy concerns related to the centralization of data for learning. Unfortunately, CML retain risks of private information leakage~\cite{boenisch2023curious,wen2022fishing}, as a passive adversary can infer sample membership~\cite{shokri2017membership,zhao2021feasibility} and reconstruct training data~\cite{luo2021feature} from the shared updates of honest users. Moreover, DL systems are especially susceptible to said privacy risk, as the architecture of DL requires all clients to share their updates with every other client for aggregation, thereby providing any client the opportunity to launch the aforementioned privacy attacks~\cite{pasquini2023security}.

% This situation prevents these clients from actively participating in future rounds of aggregation without a renewed initialization phase.  \ben{what does this previous sentence mean}
% Additionally, 

% To mitigate the communication of plain updates between clients, secret sharing, and secure aggregation algorithms can be used to safeguard the updates. Unfortunately, this introduces additional considerations whereby a single client dropping from the cryptographic process due to failing to submit their model updates—a common problem often caused by poor network connections, energy constraints, or temporary unavailability~\cite{li2020federated}. This dropped client would disrupt the cryptographic process and force a restart of the round, with enough adversaries, they effectively achieve a Denial Of Service effect on model training.

To mitigate these privacy concerns, secure aggregation is proposed to prevent the aggregator(s) from inferring private data from local model updates~\cite{ghavamipour2023federated,ghavamipour2022privacy,fereidooni2021safelearn,so2022lightsecagg}. However, implementing secure aggregation faces implementation challenges because of client dropout (as depicted in Figure~\ref{DVvsFL}). The issue of Dropout arises when a client fails to submit its model updates or cannot participate in the secure aggregation procedure due to poor network connections, energy constraints, or temporary unavailability. Such dropout can disrupt the secure aggregation process, necessitating a restart of the round, potentially leading to a Denial of Service (DoS) effect on model training.

% the implementation of secure aggregation brings about additional challenges. A prevalent issue is when a single client fails to submit its model updates due to poor network connections, energy constraints, or temporary unavailability, leading to their dropout from the cryptographic process~\cite{li2020federated}. Such a dropout can disrupt the cryptographic process, requiring a restart of the aggregation round. If enough adversaries exploit this vulnerability, it could effectively lead to a Denial of Service (DoS) effect on model training.

% Moreover, in collaborative learning systems, it is crucial to address the issue of clients dropping out during the aggregation process or failing to submit their model updates. This common problem, often caused by poor network connections, energy constraints, or temporary unavailability~\cite{li2020federated}, poses a significant challenge. The challenge lies in either recovering or excluding the dropped client's model update without compromising the integrity of other clients' updates or derailing the current aggregation round. Therefore, finding a way to ensure efficient processing in decentralized learning environments, which supports the easy recovery or exclusion of dropout clients, remains an area that needs comprehensive exploration.

% \ben{to add a sentences outlining why DL is insufficient and we need secure dl}

In this work, we introduce three novel secret sharing-based protocols designed for privacy-preserving aggregation within decentralized learning environments. These protocols enable dropout-resilient aggregation, merging clients' local models into a global model while ensuring the privacy of individual data contributions. Our approaches not only effectively mitigate privacy risks posed by 'honest but curious' clients but also uphold the integrity of the learning process despite varying levels of client participation, thereby enhancing dropout resilience. The underlying techniques are inspired by their proven effectiveness in federated learning. However, our contribution lies in adapting and refining these techniques to address the unique challenges of secure decentralized learning. This adaptation extends their applicability and enhances the privacy and security framework of the learning process, especially by incorporating mechanisms for dropout resilience. To validate the effectiveness of our proposed protocols, we conducted a comprehensive evaluation across multiple datasets, demonstrating their robustness and scalability in diverse learning scenarios.

% The underlying techniques draw inspiration from their proven effectiveness in federated learning; however, our contribution lies in their adaptation and refinement to meet the distinct challenges of secure decentralized learning. This adaptation extends their applicability and strengthens the privacy and security framework of the learning process, particularly by incorporating mechanisms for dropout resilience.

\section{Related work}

Research on privacy-preserving and dropout-resilient FL is extensive, yet studies tackling similar issues in DL are notably fewer~\cite{pasquini2023security}. In this respect, \ppdl stands out as the first to directly address dropout challenges in the realm of privacy-enhanced DL.
%, marking a significant advancement in the field.

In the context of FL, various strategies have been developed for performing privacy-preserving aggregation with a limited outlook on resiliency to client dropout. Most existing works focus on combining federated training with Differential Privacy (DP)~\cite{dwork2006differential}, Homomorphic Encryption (HE)~\cite{gentry2009fully}, and secure Multi-Party Computation (MPC)~\cite{canetti1996adaptively}. Since dropout has an impact on the privacy protection mechanism employed in the aggregation method,
highlighting the differences between these methodologies   
in terms of privacy guarantees, efficiency, and 
the implementation challenges is important.

Approaches employing homomorphic encryption result in significant performance overheads~\cite{truex2019hybrid}. This technique requires all clients to either share a common secret key or hold shares of a secret key~\cite{aono2017privacy,zhang2020batchcrypt}. These shares must be generated through complex multi-party protocols or distributed by a trusted third party~\cite{truex2019hybrid}, which results in prohibitive overheads for large-scale CML systems~\cite{zhang2022homomorphic}. Moreover, supporting dropout resilience in conjunction with homomorphic encryption poses an additional challenge, as it requires sophisticated mechanisms to handle the potential discontinuity in participation without compromising the security or the integrity of the aggregated model updates~\cite{tian2022distributed}.

%Such HE-based methods introduce prohibitive overheads for large-scale CML systems due to the expensive protocols needed by the underlying threshold cryptosystems~\cite{zhang2022homomorphic}. 
As a more scalable and dropout-resistant alternative, MPC has been proposed~\cite{bonawitz2017practical}. In this model, users' data is masked, with the seeds for generating these masks securely distributed among the users through a threshold secret-sharing scheme to manage client dropouts. An improvement on this approach, \cite{bell2020secure} replaces the complete communication graph with a \(k\)-regular graph (i.e., every vertex having a degree of $k$), reducing communication overhead. A DP-based FL method has been introduced where noise is added to each client's locally trained model before aggregation~\cite{geyer2017differentially}, effectively masking the data distribution of the clients. However, this method requires careful balancing between the level of privacy maintained and the impact on the model's performance due to perturbation.

Guo et al. \cite{guo2021topology} pioneered the application of differential privacy in decentralized learning through the introduction of LEASGD, which enhances a stochastic gradient descent algorithm with privacy protection and noise reduction strategies. While this method effectively maintains privacy at the model level, there is a competing trade-off in the utility of the model. Further advancing the field, Bellet et al. \cite{bellet2018personalized} devised a fully decentralized algorithm for training personalized ML models, embedding differential privacy to safeguard individual data privacy. Further, Xu et al.~\cite{xu2022privacy} presents \(D^2\)-MHE, a secure and efficient framework for decentralized deep learning that employs advanced homomorphic encryption techniques for the private updating of gradients. This method notably decreases communication complexity and provides robust data privacy safeguards. It emphasizes training models in an encrypted format, leveraging sophisticated cryptographic schemes such as Brakerski-Fan-Vercauteren (BFV) and Cheon-Kim-Kim-Song (CKKS). Despite its advantages in guaranteed protections of client updates, \cite{xu2022privacy} remains burdened by inefficient computation and does not address functional implications of client dropout in deployment as we do with \ppdl.

\section{Preliminaries}

In this section, we provide the necessary background on decentralized learning, and introduce the building blocks used to provide privacy of updates and resilience against dropouts in \ppdl.

\begin{figure}[t]
\centering
  \includegraphics[width=\linewidth]{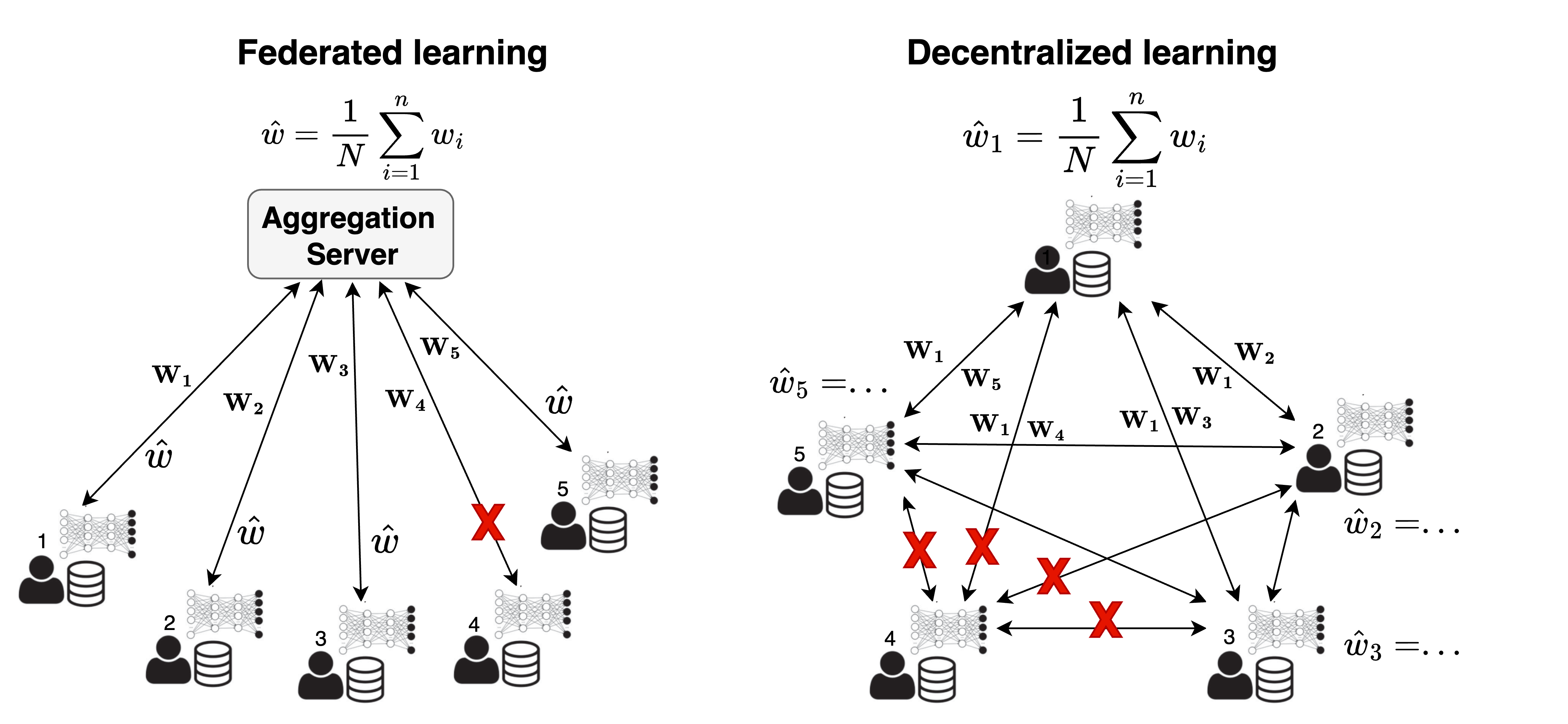}
\caption{Collaborative learning system in a federated (left) and decentralized fashion when Client 4 has dropped out.}  \label{DVvsFL}
\end{figure}

\subsection{Learning in a Decentralized Setting}\label{dc}

In a decentralized learning framework, each participant, denoted as \( c \) within the user subgroup \( C \), establishes communication links directly with a designated subset referred to as their neighbors \( \mathbf{N}(c) \). These links can be either static, established at the outset, or dynamic, subject to change over time. The ensemble of users forms an undirected graph represented by 
% \( G=\{C, \cup_{c \in C} \mathbf{N}(c)\} \), 
\( G=(C, \cup_{c \in C,~ c' \in \mathbf{N}(c)} (c, c')) \) ,
where the vertices represent the users and the edges denote the connections among them~\cite{lian2017can}.

In this environment, each participant \( c \) possesses a unique dataset \(D_i=\left\{\left(x_i, y_i\right)\right\}_i\) that is sampled from an undisclosed distribution \( \xi_i \). When these individual datasets are amalgamated, they form a comprehensive global dataset \( D \) characterized by the distribution \( \xi \). Initially, every participant is equipped with a common set of initial model parameters, labeled as \( w^0 \).

The objective of the training endeavor is to discover the optimal set of parameters, represented by \( \mathbf{w}^* \), for the machine learning model. \( \mathbf{w}^* \) aims to minimize the expected loss function over the global dataset \( D \).

\[
\mathbf{w}^* = \underset{\theta}{\arg \min } \frac{1}{|\mathbf{N}(c)|} \sum_{n_i \in \mathbf{N}(c)} \underbrace{\mathbb{E}_{s_i \sim D_i}\left[\mathcal{L}\left(\theta ; s_i\right)\right]}_{\mathcal{L}_i}
\]

For every client \( n_i \) within the network, \( \mathbb{E}_{s_i \sim D_i}\left[\mathcal{L}\left(\theta ; s_i\right)\right] \) determines the expected loss associated with the dataset \( D_i \) of that specific client, with \( s_i \) representing a sample drawn from \( D_i \). Consequently, the goal of this formulation is to identify the model parameters \( \theta \) that lead to the minimization of the average expected loss across all clients in the network.

\begin{algorithm}[ht]
\SetKwInOut{Input}{Input}
\SetKwInOut{Output}{Output}
\DontPrintSemicolon

\caption{Decentralized Learning Protocol}
\label{alg:DecentralizedLearning}

\Input{
    Initial model parameters $w_c^0$ for $c \in C$ \\
    User local training data: $X_c$ for $c \in C$
}
\For{$t \in [0, 1, \ldots]$}{
    \textbf{Local optimization step:}\;
    \For{$c \in C$}{
        Sample $x_c^t$ from $X_c$\;
        Update parameters: $w_c^{t+\frac{1}{2}} = w_c^t - \eta \nabla_{w_c^t}(x_c^t, w_c^t)$\;
    }
    \textbf{Communication with neighbors:}\;
    \For{$c \in C$}{
        \For{$u \in \mathbf{N}(c) \setminus \{c\}$}{
            Send $w_c^{t+\frac{1}{2}}$ to $u$\;
            Receive $w_c^{t+\frac{1}{2}}$ from $u$\;
        }
    }
    \textbf{Model updates aggregation:}\;
    \For{$c \in C$}{
        $w_c^{t+1} = \frac{1}{|\mathbf{N}(c)|} \sum_{c \in \mathbf{N}(c)} w_c^{t+\frac{1}{2}}$\;
    }
}
\end{algorithm}

We summarize the decentralized training of ML models in Algorithm \ref{alg:DecentralizedLearning}. The process unfolds through a sequence of stages until a specified termination point (e.g., convergence). Initially, each client \(c\) performs gradient descent on their unique model parameters, leading to the generation of an interim model update, denoted as \(w_c^{t+1/2}\). Subsequently, these clients exchange their interim model updates \(w_c^{t+1/2}\) with their network neighbors (\(\mathbf{N}(c)\)), while also receiving updates from these neighbors.

Following the exchange, clients aggregate the updates received from their neighbors with their own. This aggregation typically involves averaging the updates to modify their local state, expressed mathematically as \(w^{t+1} = \frac{1}{|\mathbf{N}(c)|} \sum_{c \in \mathbf{N}(c)} w_c^{t+1/2}\) that signifies the collective adjustment based on the aggregated updates.

In this study, we assume that the clients exhibit full connectivity, which implies that every client is directly connected to all other clients within each subgroup. This full connectivity framework ensures that every client can directly exchange information with every other client in its subgroup, facilitating a comprehensive and synchronous training process across the network.

\subsection{Shamir's Secret Sharing}

Shamir's Secret Sharing (SSS) Scheme~\cite{shamir1979share} allows a secret \(s\) to be divided into \(n\) pieces, known as shares, using a \(t\)-out-of-\(n\) scheme. Any group of \(t\) or more shares can be used to reconstruct the secret \(s\). The process is as follows:
\begin{itemize}[leftmargin=15pt]
    \item To generate shares from a secret \(s\), by using the function \\ \(Ss.
    Share(s, t, n)\), \(t-1\) random positive integers \(a_1, \ldots, a_{t-1}\) are selected from a finite field \(\mathbb{Z}_P\), with \(a_0 = s\) representing the secret. Let \(f(x) = a_0 + a_1x + a_2x^2 + \cdots + a_{t-1}x^{t-1} \mod P\) be a polynomial with coefficients from \(\mathbb{Z}_P\), where \(\mathcal{U} = \{u_i | u_i \in [1, P]\}\) is a set of \(n\) random numbers and \(0 < t \leq n < P\). The shares are then \(\{(u_i, s_i)\}_{u_i \in \mathcal{U}} = \{(u_i, f(u_i))\}_{u_i \in \mathcal{U}}\).
    \item To reconstruct the secret \(s\) from shares using \(Ss.Recon\\
    (\{(u_i, s_i)\}_{u_i \in \mathcal{U}}, t)\), any subset of \(t\) shares can be used for the recovery of \(s\) through Lagrange interpolation. Since \(a_0 = s\), the secret is efficiently found as \(s = a_0 = f(0) = \sum_{j=0}^t f(u_j) \\
    \prod_{i=0, i \neq j}^t \frac{u_i}{u_i - u_j} \mod P\).
\end{itemize}

Moreover, SSS method inherently supports operations that resemble additive homomorphism, allowing for the direct combination of secrets. If the index \(i\) and the corresponding element \(u_i\) of any share in the set \(\{(u_i, s_i)\}_{u_i \in \mathcal{U}}\) for a given secret \(s\) match the index and element \(v_i\) in another set \(\{(v_i, s'_i)\}_{v_i \in \mathcal{V}}\) for a different secret \(s'\), the secrets can be seamlessly combined. The combined secret \(s + s'\) can be efficiently reconstructed by applying the reconstruction function \(Ss.Recon\) on these merged shares.

The security of Shamir's \(t\)-out-of-\(n\) secret sharing ensures that for any two sets of shares \(U, V\) both less than threshold \(t\), the output of \(Ss.Recon(V)\) is indistinguishable from \(Ss.Recon(U)\), thereby maintaining the confidentiality of the secret \(s\).

\subsection{Diffie-Hellman Key Agreement Protocol}\label{DH}

The Diffie-Hellman (DH) key agreement protocol~\cite{hellman1976new} is structured around probabilistic polynomial-time (PPT) algorithms, as outlined below:

\textbf{Parameter Generation:} \((\mathbb{G}^\prime, g, q, H) \leftarrow \text{DH}\text{Param}(\kappa)\): Given a security parameter \(\kappa\), this algorithm generates a cyclic group \(\mathbb{G}^\prime\) of prime order \(q\) and a generator \(g\) of the group, and specifies a hash function \(H\), such as SHA-256, for hashing purposes.

\textbf{Key Generation:} \((s, g^s) \leftarrow \text{DHGen}(\mathbb{G}^\prime, g, q)\): This algorithm selects a secret key \(s\) from the integer group \(\mathbb{Z}_q\) and computes \(g^s\) as its corresponding public key.

\textbf{Key Agreement:} \(s_{a,b} \leftarrow \text{DHAgree}(s_a, g^{s_b})\): For a secret key \(s_a\) and a public key \(g^{s_b}\), generated by another party's secret key \(s_b\), this algorithm produces the shared secret key \(s_{a,b} = H((g^{s_b})^{s_a})\), leveraging the hash function \(H\) for the final key derivation.

The security of the DH key agreement protocol, when deployed to protect against honest-but-curious adversaries, relies on a fundamental cryptographic principle known as the Decisional Diffie-Hellman (DDH) assumption~\cite{boneh1998decision}. This assumption secures the protocol by rendering it computationally difficult to differentiate between instances of the DH exchange and random group elements.

\section{Problem Definition}

Recall that in a collaborative learning setting, an honest client aims to obtain an aggregated model update that integrates local models trained on other clients' private datasets. The hidden aggregation of these updates helps mitigate privacy risks to client datasets in two main ways. First, it protects individual local models from malicious scrutiny by obscuring specific details, preventing adversaries from identifying and exploiting exact data sources. Second, it reduces the impact of any single client's input by merging contributions from multiple clients into the aggregated model.

In FL, a central server collects model updates from clients before distributing the aggregated results back to them. This centralized approach ensures that only the server can access each client's individual updates. However, in DL, the absence of a central server elevates privacy risks, as updates must be shared directly between clients for aggregation~\cite{pasquini2023security}. These risks can be mitigated by implementing secure aggregation schemes.

Moreover, when secure aggregation is implemented in a collaborative learning setting, client dropout poses a significant challenge~\cite{bonawitz2017practical}. Clients may exit the aggregation process before sharing their model updates due to latency, network connectivity issues, or unexpected shutdowns. This can occur before they share their updates with any other clients or after having distributed them to everyone. Since dropout is a prevalent concern in deployed systems, designing systems resilient to client dropouts up to a certain threshold is vital. 

This work presents privacy-preserving DL algorithms that are resilient to dropout, ensuring the model remains accurate and operationally efficient.

% The dropout problem can be described as the need to maintain a minimum threshold of active participants. This threshold is based on the number of semi-honest clients plus an additional buffer to accommodate the potential loss of honest clients.  The proposed aggregation scheme should guarantees the the accurate completion of the aggregation amid fluctuating participant availability.

% Our proposed privacy-preserving methods are secured by a privacy-preserving aggregation rule that prevents any client from directly accessing the model updates of another client, instead allowing only visibility to aggregated outcomes.

\subsection{Threat model}

% Our threat model consists of a landscape dominated by non-malicious yet inquisitive clients, operating under the assumption of an honest majority. 
% While adhering to the established protocols of the system, these `honest-but-curious' clients want to glean or infer additional information about their peers through the interactions of DL. 
% Their modus operandi does not involve transmitting incorrect or manipulated data; instead, their curiosity drives them to derive as much insight as possible from the communication flows inherent in the learning process. This inclination, albeit not malicious in intent, poses a significant threat to privacy and data confidentiality within the decentralized network. 

% The potential for nodes to unpredictably dropout or exhibit intermittent inactivity adds a layer of complexity to this scenario. Such behavior challenges the integrity of the learning process and could potentially be exploited to amplify privacy risks. Therefore, our model emphasizes the need for sophisticated mechanisms that can adeptly manage the dual concerns of curbing the curiosity-driven inference attacks and ensuring the resilience of the learning process amidst the uncertainties of client dropouts, thereby safeguarding private local data and providing privacy assurances when deployed in realistic environments.

% \ben{The problem statement and Threat model feels like it's repeating things.}

% \ben{perhaps a unified one, with two clear headings a) privacy risks b) dropout}

The threat model considered in this paper encompasses a landscape dominated by non-malicious yet curious clients operating under the assumption of an honest majority. While adhering to the established protocols of the system, these 'honest-but-curious' clients want to obtain or infer additional information about their peers through the interactions of DL. Their approach does not involve transmitting incorrect or manipulated data; instead, they aim to derive as much insight as possible from the communication flows inherent in the learning process. Adding to this complexity, we consider the scenario where a semi-honest client may collude with others. This collusion aims to pool resources and shared insights to enhance their collective ability to infer private information beyond what is intended through the system's interactions.

Therefore, the threat model underscores the need for advanced mechanisms to manage individual curiosity-driven inference attacks and collaborative efforts among semi-honest clients. These mechanisms are vital for providing robust privacy assurances when deployed in realistic environments, protecting private local data against these enhanced potential privacy risks.

\section{Approaches}
\label{sec:approaches}
% The primary aim of secure aggregation is to ensure that the receiving client cannot access individual model updates, which is crucial for maintaining participant data privacy. Each client's contribution is merged during the aggregation process so that their specific model updates remain undisclosed. Only the aggregated result, a collective compilation of all updates, is revealed after a sophisticated averaging procedure.
In this section, we introduce three distinct approaches for implementing secure aggregation in DL environments, emphasizing the incorporation of dropout resilience. 

% This adaptability ensures that the aggregation process remains unaffected by client dropout, maintaining the integrity and effectiveness of the secure aggregation framework.

% \ben{how would you describe the most significant difference between the approach's? the technique novelty? the efficiency?}

\subsection{SA using Shamir's secret sharing}\label{SSS}

% A straightforward strategy to achieve dropout resilience in privacy-preserving aggregation for deep learning (DL) environments might utilize Shamir's Secret Sharing.

This method allows clients to work together to build a shared model vector, \(W\), enhancing resilience against dropout by enabling the model's reconstruction even if some clients do not contribute. Each client participates in the aggregation process by supplying vectors made up of field elements from a field of size \(q\). This arrangement aligns with the objectives of the decentralized model. The approach not only addresses the complexities of distributed computation but also introduces an effective mechanism for managing client dropout. It ensures the confidentiality and integrity of the model aggregation process, eliminating the need for a centralized server.

% We refer to the direct application of SSS to provide dropout resilience in privacy-preserving aggregation within DL settings, naive \ppdl (\ppdlbase).
% This method enables clients to collaboratively construct a shared model vector \( W \), facilitating dropout resilience by allowing the reconstruction of the model even when some clients fail to contribute. Each client contributes in the aggregation by providing vectors composed of field elements of length \( l \), which aligns with the decentralized model's objectives. This approach not only supports the distributed computation complexities but also provides a mechanism to handle client dropout effectively, ensuring the confidentiality and integrity of the model aggregation process without reliance on a centralized server. 

\begin{algorithm}
\SetKwInOut{Input}{Input}
\SetKwInOut{Output}{Output}
\DontPrintSemicolon

\caption{\ppdlbase: Secure Aggregation via Shamir's Secret Sharing for PPDL}
\label{alg:DecentralizedSSS}

\Input{
    $n$ clients $C_0, C_1, \dots, C_n$ each with a local training dataset $D_i$, for $i = 0, 1, \dots, n$ \\
    Number of global iterations $R_g$ \\
    Shamir's threshold $t$ (e.g., $t=(n/2)+1$)
}
\Output{A globally trained model \(\mathbf{w}\) for each client}
\For{$r = 1$ \KwTo $R_g$}{
    \tcp{Initialization step}
    All clients start with the same initial model \(\mathbf{w}^0\).\;
    \For{$i = 0$ \KwTo $n$}{
        Train local model: \(\mathbf{w}_i \leftarrow \text{LocalUpdate}(\mathbf{w}^{r-1}, D_i)\)\;
        
        Generate and distribute $n$ shares of \(\mathbf{w}_i\), denoted as \(\llbracket \mathbf{w}_i \rrbracket_j\) for \(j=0,1,\ldots,n\), to all clients, including itself\;
        
        Receive shares \(\llbracket \mathbf{w}_j \rrbracket_i\) from all other clients\;
        
        Aggregate own share with received shares from other clients: 
        \(\llbracket \mathbf{\bar{w}} \rrbracket_i = \sum_{\substack{j=0 \\ j \neq i}}^{n} \llbracket \mathbf{w}_j \rrbracket_i+ \llbracket \mathbf{w}_i \rrbracket_i\)\;
        
        Transmit \(\llbracket \mathbf{\bar{w}} \rrbracket_i\) to all other clients and receive \(\llbracket \mathbf{\bar{w}} \rrbracket_j\) from them\; 
        
        \If{total shares received $\geq t$}{
            Perform Shamir Reconstruction to obtain \(\mathbf{\bar{w}}\) and update global model: \(\hat{\mathbf{w}}^{r} = \frac{1}{n}\mathbf{\bar{w}}\)\;
        }
    }
    % \tcp{Proceed to the next iteration with updated global model}
}
\end{algorithm}

In this approach summarized in Algorithm~\ref{alg:DecentralizedSSS}, each client \( P_i \) generates \( n \) shares for each element of their input vector \( w_i \), adhering to a predetermined threshold \( t \). The threshold \(t\) should exceed the number of semi-honest clients to protect the integrity and confidentiality of the aggregation process. Next, each client distributes one share to the other clients while retaining a share for itself. Subsequently, clients aggregate the shares of the secret value received from others with their own and broadcast the cumulative value to all participants. Once sufficient shares of these combined values are collected, every client can reconstruct the aggregated value.

% To handle dropout clients within our synchronous decentralized learning approach, a client may dropout either before broadcasting their model updates to any clients or after broadcasting their model updates shares to all users. In this context, ensuring that the aggregation process can tolerate the dropout of clients up to a specific limit is critical. Specifically, the system must be designed so that the total number of dropouts does not cause the number of participating clients to fall below the threshold \(t\), determined by the sum of semi-honest clients and an additional margin to protect against the dropout of honest clients. This protection ensures that the aggregation process remains resilient to dropouts, thereby upholding the system's overall integrity and confidentiality despite the potential inconsistencies in participant availability.

Although this approach inherently tolerates dropouts, the inefficiency of traditional Shamir's Secret Sharing in a decentralized setting primarily stems from the extensive number of secret values that need to be shared. These secret values correlate directly with the number of model weights, resulting in excessive communication and computation overheads.

Improvements to this inefficiency can be achieved by utilizing Packed Shamir Secret Sharing~\cite{franklin1992communication}, a drop-in adaptation of Shamir's Secret Sharing in all three of our proposed approaches. Packed Shamir Secret Sharing significantly improves communication and computation efficiency by including multiple values within a single share. The packed secret sharing method has the capacity to split $k$ secret values into $n$ shares while requiring at least $t+k$ shares to recover the secret. Despite these improvements, this proposed scheme's naive application, still  observes the vector size for secret sharing still correlating with the model's parameter size. The packing only lowers the per-client communication burden, particularly with large-scale models.

\subsubsection{Security Analysis}\label{SSS_secure}
In this part, we explore the security features of the \ppdlbase algorithm in detail:

\begin{theorem}
The \(\ppdlbase\) protocol, which utilizes SSS to distribute model updates as \(n\) shares within a finite field \(\mathbb{F}_q\) among clients, ensures security against semi-honest adversaries and exhibits resilience to client dropout in the honest majority setting.\end{theorem}

\begin{proof}

In \ppdlbase protocol, Shamir's Secret Sharing (SSS) ensures the confidentiality of model updates. Each element of a model update \(s\) is divided into \(n\) shares within a finite field \(\mathbb{F}_q\) (assuming \(n < q\)), according to \(A(x) = \sum_{i=0}^{t-1} a_i x^i\), where \(a_0 = s\) and the coefficients \(a_i\) are randomly and uniformly chosen from \(\mathbb{F}_q\). This approach utilizes unique, non-zero points \(x_i\) in \(\mathbb{F}_q\), ensuring each share \(v_i = A(x_i)\) contains part of the secret without fully disclosing it.

The security of SSS stems from the fact that a polynomial of degree \(t-1\) can only be uniquely determined by any \(t\) distinct evaluations (\(v_i\)), making it mathematically infeasible to reconstruct \(s\) from fewer than \(t\) shares. This is due to there being \(q^{t-1}\) potential polynomials that could correspond to any given set of \(t-1\) shares, preserving the secrecy of \(s\).

This mechanism also underpins the protocol's resilience to client dropout, as the threshold \(t\) can be adjusted to maintain confidentiality despite partial participation. In essence, for any subset of \(t-1\) shares, the distribution of potential secrets remains uniform, ensuring no leakage of information about \(s\).

For reconstruction, the protocol employs Lagrange polynomials \(L_i\) for each share \(i\), calculated as \(L_i = \frac{\prod_{j\neq i}(X-x_j)}{\prod_{j\neq i}(x_i-x_j)}\), where \(L_i(x_i) = 1\) and \(L_i(x_j) = 0\) for all \(j \neq i\). This enables the secure computation of \(A(0) = s = \sum_{i=1}^t v_i L_i(0)\) from precisely \(t\) shares.

\end{proof}

\subsection{SA using LWE-Based Masking}

As we mentioned previously, packed secret sharing significantly improves the communication efficiency of \ppdlbase, however, the quantity of values that need to be secret shared still directly correlates to the size of the model update vector. 
To address this limitation in our second proposal, we draw inspiration from the work of Stevens et al.~\cite{stevens2022efficient} and propose the use of a dropout-resilient Differentially Private (DP) secure protocol grounded with the Learning With Errors (LWE) assumption for a decentralized learning setting. This second approach which we call \ppdlerror, each client generates a one-time pad of equal length to their model parameters and broadcast it to all other clients. Utilizing SSS, clients collaboratively sum their masks and then subtract each client's mask from the aggregated model updates value, thereby enhancing both the efficiency and security of the process.

\begin{algorithm}
\SetKwInOut{Input}{Input}
\SetKwInOut{Output}{Output}
\DontPrintSemicolon

\caption{\ppdlerror: Secure Aggregation via LWE-Based Masking for PPDL}
\label{alg:lwe}

\Input{
    $n$ clients $C_0, C_1, \dots, C_n$ each with a local training dataset $D_i$, for $i = 0, 1, \dots, n$ \\
    Number of global iterations $R_g$ \\
    Public \(m \times n\) matrix \(A\) over finite field \( \mathbb{F}_q \)
}
\Output{A globally trained model \(\mathbf{w}\) for each client}
\For{$r = 1$ \KwTo $R_g$}{
    \tcp{Initialization step}
    All clients start with the same initial model \(\mathbf{w}^0\).\;
    \For{$i = 0$ \KwTo $n$}{
        Train local model: \(\mathbf{w}_i \leftarrow \text{LocalUpdate}(\mathbf{w}^{r-1}, D_i)\)\;
        
        Generate secret vector \(\mathbf{s}_i\) of size \(n\)\;
        Generate random noise vector \(\mathbf{e}_i\) of size \(m\)\;
        Compute masking vector \(\mathbf{b}_i \leftarrow A\cdot\mathbf{s}_i + \mathbf{e}_i\)\;
        Mask model parameters: \(\mathbf{h}_i \leftarrow \mathbf{w}_i + \mathbf{b}_i\)\;
        Transmit \(\mathbf{h}_i\) to all clients\;
        Receive \(\mathbf{h}_j\) from all other clients\;
        
        Distribute \(\llbracket \mathbf{s}_i \rrbracket\) to all clients\;
        Collect \(\llbracket \mathbf{s}_j \rrbracket\) from all clients\;
        
        \tcp{Using Secret Sharing for aggregation}
        Compute \(\mathbf{s}_{\text{sum}} \leftarrow \sum_{j} \mathbf{s}_j\)\;
        
        Aggregate masked model parameters: \(\mathbf{H}_{\text{sum}} \leftarrow \sum_{j} \mathbf{h}_j\)\;
        Compute aggregated model update: \(\mathbf{W}_{\text{sum}} \leftarrow \mathbf{H}_{\text{sum}} - A\cdot\mathbf{s}_{\text{sum}}\)\;
        
        Update global model: \(\hat{\mathbf{w}}^{r} \leftarrow \frac{1}{n}\mathbf{W}_{\text{sum}}\)\;
    }
    % \tcp{Proceed to the next iteration with updated global model}
}
\end{algorithm}

As illustrated in Algorithm~\ref{alg:lwe}, clients share a public \(m \times n\) matrix \(A\), where each element of the matrix is from the finite field \( \mathbb{F}_q \). Each client \(i\) independently generates a secret vector \(s_i\) of size \(n\) and an error vector \(e_i\) of size \(m\), both derived from the same distribution. The client then computes a vector \(b_i = A\cdot s_i + e_i\), where \(A\) 
% \(b_i = A \cdot s_i + e_i\)
is the public \(m \times n\) matrix. This forms an LWE sample with the pair \((A, b_i)\), where \(b_i\) effectively serves as a one-time pad to mask the client's model update \(w_i\), resulting in a masked vector \(h_i = w_i + b_i\).

Next, clients transmit \(h_i\) to other clients and subsequently compute the aggregation of the received vectors, \(h_{\text{sum}}\), through vector addition. The \(h_{\text{sum}}\) is composed of \(w_{\text{sum}}+A s_{\text{sum}}+e_{\text{sum}}\), where \(e_{\text{sum}}\) is the added noise that satisfies the \((\varepsilon, \delta)\)-DP criterion, and \(w_{\text{sum}} + e_{\text{sum}}\) forms the noisy aggregated value of the model updates from other clients. Therefore, to eliminate \(A s_{\text{sum}}\) from \(h_{\text{sum}}\), clients only need to compute \(s_{\text{sum}}\). To determine \(s_{\text{sum}}\) without disclosing the individual \(s_i\) values, clients collaboratively use the SSS protocol, as introduced in Section~\ref{SSS}. Finally, each client, using the shared matrix \(A\) and \(s_{\text{sum}}\), computes \(h_{\text{sum}}\) as the aggregated value.

To ensure this protocol's effectiveness, two key points are crucial for balancing privacy and performance. First, protecting the confidentiality of individual error vectors \(e_i\) is essential, leveraging the LWE assumption for privacy. By aggregating these vectors into a cumulative noise \(e_{\text{sum}}\) via discrete Gaussians, the protocol achieves differential privacy, offering a precise yet privacy-aware estimation of \(w_{\text{sum}}\). Second, the dimension of the secret vector \(S\) plays a vital role in enhancing security. Stevens et al.~\cite{stevens2022efficient} emphasize the importance of keeping \(S_i\)'s length at a minimum of 710 for robust security, regardless of the secret input vector \(w_i\)'s dimensionality. These elements are critical in creating a strong decentralized learning framework that effectively addresses privacy and efficiency challenges. Beyond the specific focus of adversarial robustness, adversarial training has been shown to increase the generalization, and thus the model performance on regular binaries too.

\subsubsection{Security Analysis}
In this subsection, we discuss the security provided by \ppdlerror algorithm in detail:

\begin{theorem}

In settings with an honest majority, the \(\ppdlerror\) protocol is secure against semi-honest adversaries by leveraging the LWE assumption for the encryption of model updates and utilizing SSS for the secure aggregation of clients' secret vectors. This protocol exhibits resilience to client dropout, ensuring the confidentiality and integrity of the DL process.

\end{theorem}

\vspace{-2mm} 

\begin{proof}
The cornerstone of the \(\ppdlerror\) protocol's security lies in its employment of the LWE assumption, a principle positing that distinguishing between genuine LWE samples \((A, A \cdot s_i + e_i)\) and random pairs from the same distribution is computationally infeasible for any polynomial-time adversary. Here, \(A\) represents a publicly known matrix, \(s_i\) is a secret vector selected by the client, and \(e_i\) is an error vector with elements drawn from a discretized Gaussian distribution. The distribution's parameters are chosen to maintain the encryption's integrity while enabling efficient decryption by legitimate entities.

The intractability of the LWE problem, rooted in the computational complexity of solving certain lattice problems deemed hard for quantum computers, underpins the protocol’s resistance to quantum attacks. Without exact knowledge of \(s_i\) and \(e_i\), decryption of the encrypted model update \(h_i = w_i + A \cdot s_i + e_i\) by adversaries is rendered impractical. This encryption model ensures that data confidentiality is upheld, even in the face of semi-honest clients who follow the protocol yet attempt to extract additional information.

Further enhancing the protocol’s security is the application of SSS for aggregating the secret vectors \(s_i\) into a collective sum \(s_{\text{sum}}\). The SSS mechanism necessitates a predetermined number of shares to reconstruct any individual secret vector, thereby precluding semi-honest clients from independently accessing or inferring another's secret vector. The collaborative computation of \(s_{\text{sum}}\), essential for decrypting the combined encrypted vectors \(h_{\text{sum}}\), intrinsically protects against the potential disclosure of sensitive data to semi-honest participants during the aggregation phase.

\end{proof}

\subsection{SA with pairwise masking}

% The third proposed approach involves each node in the network masking its data before sharing, ensuring that the aggregated data remains private and secure throughout the process. At the core of this technique is the use of sophisticated cryptographic methods to maintain data integrity and confidentiality. A notable implementation of this concept in Federated Learning seeting can be found in the protocol proposed by Bonawitz et al. \cite{bonawitz2017practical} , which has been specifically designed to address the unique challenges of secure aggregation in decentralized settings.

% This protocol effectively reduces per-client communication costs to \(O(n + l)\) by leveraging pairwise masking. The process is based on two clients, say \(A\) and \(B\), using complementary masks on their input vectors. Specifically, client \(A\) adds a random mask to their data, while client \(B\) subtracts the same mask from theirs. When these masked vectors are combined, the mask is nullified, revealing the sum of the original vectors without compromising the privacy of individual data, thanks to the randomness of the masks.

Building on the principles of the previous approach, the third proposed protocol emphasizes masking clients' model parameters before transmission to other clients. This method introduces a novel aspect where clients engage in pairwise sharing of a seed for a pseudo-random number generator (PRNG)~\cite{blum2019generate,narayanan2008robust}, allowing each client to produce a shared vector tailored to the desired size locally. This technique addresses the unique challenges of secure aggregation by incorporating principles akin to those explored by Bonawitz et al. \cite{bonawitz2017practical} in their work on Federated Learning. It offers a refined solution for enhancing privacy and security in decentralized learning environments.

In this protocol, each client \(C_i\) generates a unique pairwise mask for their vector relative to every other neighbouring client \(C_j\). The collective result of combining these masked vectors \(\mathbf{y}_i\) is the aggregation of the original inputs, effectively concealing the individual vectors \(\mathbf{w}_i\). The generation of these masks is facilitated by a PRNG, which uses a result derived from each client's public keys, exchanged during the initial steps of the protocol. This mechanism guarantees a secure and agreed-upon method for mask generation, ensuring the confidentiality of the data during the aggregation process.

\[
y_i = w_i + \sum_{j: i<j} \text{PRNG}(s_{i, j}) - \sum_{j: i>j} \text{PRNG}(s_{j, i})
\]

The Diffie-Hellman key exchange protocol provides a robust method for generating a common seed for a PRNG, essential for creating pairwise masks in secure communications. After generating and exchanging public keys \(g^{a_i} \mod q\), each client uses their private key \(a_i\) in conjunction with the received public keys from other clients to generate a shared secret seed \(s_{i, j}\). This seed initializes the PRNG, allowing each pair of clients to agree on a complementary sequence of random numbers for mask generation, facilitating secure and pairwise synchronized communications across the network.

% Client A computes the shared secret as \( (g^b)^a \) mod \( q \), and Client B computes it as \( (g^a)^b \) mod \( q \). Due to the properties of exponentiation in modular arithmetic, these computed values are identical, resulting in a shared secret key \( g^{ab} \) mod \( q \).

% Let's denote the private key of Client A as \(a\) and that of Client B as \(b\). The public keys are derived from these private keys using a common base \(g\) and a prime number \(q\) (representing the group order). Hence, the public key of Client A is \(g^a \mod q\), and the public key of Client B is \(g^b \mod q\).

% These public keys are exchanged between Client A and Client B without risk of compromise. Following the exchange, each client uses the received public key and their own private key to compute a shared secret key. Client A computes the shared secret as \( (g^b)^a \) mod \( q \), and Client B computes it as \( (g^a)^b \) mod \( q \). Due to the properties of exponentiation in modular arithmetic, these computed values are identical, resulting in a shared secret key \( g^{ab} \) mod \( q \).
\begin{algorithm}
\SetKwInOut{Input}{Input}
\SetKwInOut{Output}{Output}
\DontPrintSemicolon
\caption{\ppdlpairs: Secure Aggregation with Pairwise Masking and Diffie-Hellman (DH) Key Exchange}
\label{alg:DecentralizedBonawitz}

\Input{
    \(n\) clients \(C_0, C_1, \dots, C_n\) each with a local training dataset \(D_i\), for \(i = 0, 1, \dots, n\) \\
    Number of global iterations \(R_g\) \\
    Common base \(g\) and prime \(q\) for DH key exchange
}
\Output{A globally trained model \(\mathbf{w}\) for each client}
\For{\(r = 1\) \KwTo \(R_g\)}{
    \tcp{Initialization step}
    All clients start with the same initial model \(\mathbf{w}^0\).\;
    \For{\(i = 0\) \KwTo \(n\)}{
        Train local model: \(\mathbf{w}_i \leftarrow \text{LocalUpdate}(\mathbf{w}^0, D_i)\)\;

        % Diffie-Hellman Key Exchange
        Generate private key \(a_i\) and compute public key \(g^{a_i}\)\;
        Exchange public keys \(g^{a_i}\) with other clients\;
        
        % Secret sharing of private key
        Distribute shares of \(a_i\) using SSS to each client\;

        % Generation of masks using shared secret and secret sharing
        \For{each client \(C_j\) where \(j \neq i\)}{
            Receive public key \(g^{a_j}\) from client \(C_j\)\;
            Compute shared secret \(s_{i, j} = (g^{a_j})^{a_i} \)\; % This is the g^ab calculation
            Generate pairwise masks PRNG(\(s_{i, j}\))\;
        }
        Generate personal mask \(b_i\) using PRNG\;
        \(\mathbf{y}_i \leftarrow \mathbf{w}_i + b_i\)\;
        
        \For{each client \(C_j\) where \(j \neq i\)}{
            \eIf{\(i < j\)}{
                \(\mathbf{y}_i \leftarrow \mathbf{y}_i + \text{PRNG}(s_{i, j})\)\; % Pairwise masks addition
            }{
                \(\mathbf{y}_i \leftarrow \mathbf{y}_i - \text{PRNG}(s_{j, i})\)\; % Pairwise masks subtraction
            }
        }
        % Broadcasting aggregated masked update
        Broadcast \(\mathbf{y}_i\) to all other clients\;
    }
    
    % Collecting and aggregating all received masked updates
    Initialize an empty vector \(\mathbf{V}\) for aggregation\;
    \For{each received masked vector \(\mathbf{y}_j\) from clients}{
        \(\mathbf{V} \leftarrow \mathbf{V} + \mathbf{y}_j\)\;
    }
    
    % Handling client dropout
    \If{a client \(C_k\) drops out}{
        % Collect shares of \(a_k\) to reconstruct the private key
        Collect shares of \(a_k\) from remaining clients to reconstruct \(a_k\)\;
        \For{each client \(C_j\) where \(j \neq k\)}{
            % Recompute the shared secret and the pairwise masks for dropped client
            Compute shared secret \(s_{j, k}\) and \(s_{k, j}\) using \(a_k\)\;
            Adjust \(\mathbf{V}\) by removing the mask contributions;
        }
    }
    
    % Final aggregation and global model update
    \(\hat{\mathbf{w}}^{r} \leftarrow \frac{1}{n}\mathbf{V}\); % Update to match the representation style
}
\end{algorithm}

Similar to earlier approaches, we tackle the issue of clients dropping out after broadcasting their mask vector. A straightforward approach is having the remaining clients remove the dropout mask and resend the updated vector. However, this brings us a new challenge: more clients may drop out during this recovery phase before sharing their seeds, requiring additional recovery phases for these new dropouts. As a result, this approach could trigger a cascade of recoveries, possibly equaling the total number of clients and complicating the process.

To circumvent this issue, we employ the threshold secret sharing protocol introduced in Section~\ref{SSS}, where users distribute a share of their private key to all other clients. In the case of a client dropping out after sending its masked model update vector, remaining client utilizes the SSS to reconstruct the dropout's private key. This enables the recalculation and removal of the dropout's mask from the aggregated value. While this solution necessitates additional steps for seed recovery and mask recalibration, it ensures the integrity of the aggregation process is maintained, even in the event of client dropouts.

% In contrast to Federated Learning where client dropout can pose a serious issue, in decentralized learning environments, the protocol is designed to mitigate the impact of dropouts. In these scenarios, if a client drops out during the protocol's execution, the remaining clients can adjust by removing the mask that was generated by the dropout user. This ensures the accuracy of the aggregate sum. The key factor in this approach is that each node should have at least two neighbors. 

\begin{table*}[ht]
\resizebox{0.7\linewidth}{!}{
\begin{tabular}{|c|c|c|c|c|c|c|}
\hline
\# of users          & Dataset       & DL         & \ppdlbase     & \ppdlerror                        & \ppdlpairs   & DLDP                       \\ \hline
\multirow{4}{*}{50}  & MNIST         & $96.56\%$  & $96.52\%$  & $89.12 \%(\epsilon=9.25)$  & $96.22\%$  & $89.85 \%(\epsilon=9.25)$  \\
                     & CIFAR-10      & $74.46 \%$ & $73.98 \%$ & $62.12 \%(\epsilon=11.31)$ & $73.9 \%$  & $62.74 \%(\epsilon=11.31)$ \\
                     & SVHN          & $92.92\%$  & $92.61\%$  & $83.03\%(\epsilon=11.00)$  & $91.82\%$  & $83.07\%(\epsilon=11.00)$  \\
                     & Fashion-MNIST & $87.48\%$  & $87.22\%$  & $83.01\%(\epsilon=11.00)$  & $86.92\%$  & $83.19\%(\epsilon=11.00)$  \\ \hline
\multirow{4}{*}{100} & MNIST         & $96.24 \%$ & $96.11\%$  & $88.60 \%(\epsilon=9.33)$  & $96.11 \%$ & $89.12 \%(\epsilon=9.33)$  \\
                     & CIFAR-10      & $71.24 \%$ & $71.02 \%$ & $62.21 \%(\epsilon=11.58)$ & $71.20 \%$ & $62.56 \%(\epsilon=11.31)$ \\
                     & SVHN          & $92.12\%$  & $92.03\%$  & $82.47\%(\epsilon=11.00)$  & $91.23\%$  & $82.97\%(\epsilon=11.00)$  \\
                     & Fashion-MNIST & $87.13\%$  & $87.02\%$  & $82.21\%(\epsilon=10.610)$ & $87.12\%$  & $82.91\%(\epsilon=10.61)$  \\ \hline
\end{tabular}
}
\caption{Accuracy comparison of different aggregation method in Decentralized learning }
\label{tab:accuracy_compare}
\end{table*}

\subsubsection{Security Analysis}

In this subsection, we detail the security measures implemented by the \ppdlpairs algorithm:

\begin{theorem}
\ppdlpairs protocol, under the Decisional Diffie-Hellman (DDH) assumption, establishes a secure protocol for exchanging keys and masking model updates. This protocol ensures the confidentiality and integrity of decentralized learning processes, bolstered by the indistinguishability of shared secrets and enhanced with Shamir's Secret Sharing (SSS) for dropout resilience.
\end{theorem}

\begin{proof}
Given a cyclic group \(\mathbb{G}\) of order \(q\) with generator \(g\), the DDH assumption (as described in~\ref{DH}) states that for any polynomial-time adversary, distinguishing between tuples \((g^a, g^b, g^{ab})\) and \((g^a, g^b, g^c)\), where \(c\) is chosen uniformly at random from \(\mathbb{Z}_q\), is computationally hard. This assumption is crucial for securing the key exchange phase, ensuring shared secrets \(g^{ab}\) are indistinguishable from random and protected against eavesdropping.

Following key exchange, the protocol uses a PRNG seeded with Diffie-Hellman derived shared secrets \(s_{ij} = g^{a_i a_j}\) to generate cryptographic masks for obfuscating model parameters \(w\). The security of this masking relies on the assumption that the hash function \(H\), applied to \(s_{ij}\) for randomness extraction, produces a uniformly random string that serves as a secure PRNG seed. These masks ensure confidentiality through computational secrecy, predicated on the DDH problem's hardness.

In the aggregation phase, clients compile masked model updates without exposing individual contributions, preserving privacy and integrity. This is achieved by each client \(C_n\) aggregating the received masked updates \(\sum y_i^{(w)}\), where \(y_i^{(w)} = w_i + \sum m_{ij}\), leveraging the computational indistinguishability of the masks from random due to the secure PRNG process.

Moreover, integrating Shamir's Secret Sharing (SSS) with a defined threshold enhances the protocol's security against dropout. It leverages the uniform randomness and indistinguishability of shares to preserve confidentiality and integrity, even amid client dropouts. This strategic addition ensures the aggregation process remains secure despite network variability, reinforcing the protocol's robust security framework and its efficiency in decentralized settings.

\end{proof}

\section{Evaluation}
In this section, we conduct a set of  experiments to demonstrate the robustness and efficiency of our proposed protocols.

\subsection{Experimental Setup}

We implemented the protocols in Python, utilizing the Pytorch framework for model training~\cite{paszke2019pytorch}\footnote{The implementation will be made publicly available upon acceptance.}. All experiments were conducted on a compute cluster with an Intel Xeon Gold 6150, Nvidia V100 GPU, and 128GB RAM in the local network. 

\subsubsection{Datasets} For our assessment, we utilized four widely recognized datasets in deep learning: MNIST~\cite{lecun1998mnist}, Fashion-MNIST~\cite{xiao2017fashion}, SVHN~\cite{netzer2011reading}, and CIFAR-10~\cite{krizhevsky2009learning}, all partitioned into subsets that are independent and identically distributed (IID). We divided the training data into smaller subsets (generally in equal sizes) and distributed these among the clients at random. The MNIST dataset is comprised of 60,000 training images and 10,000 testing images of handwritten digits in grayscale, with each image measuring \(28 \times 28\) pixels. Fashion-MNIST offers a dataset for a 10-class fashion item classification challenge, including 60,000 training images and 10,000 testing images of fashion items. Derived from Google Street View images of house numbers, the SVHN dataset contains 99,289 color images in ten categories, with a training set of 73,257 images and a testing set of 26,032 images, each standardized to \(28 \times 28\) pixels. CIFAR-10 presents a varied set of 50,000 training and 10,000 testing color images distributed over ten unique categories, with each image depicting one category.

\subsubsection{Model}

We implemented tailored model architectures to address the unique characteristics of our datasets: a CNN for CIFAR-10 and SVHN and a two-layer MLP for MNIST and Fashion-MNIST. The CNN architecture features two convolutional layers with a kernel size of \(3 \times 1\) per layer, followed by three fully connected layers, with 384 neurons in each hidden layer and 10 neurons in the output layer. This structure is designed to effectively process and learn from the complex image data presented by CIFAR-10. The model is optimized with a learning rate of 0.002 and a batch size of 128, incorporating Group Normalization (GN) alongside max pooling and dropout rates between 0.2 to 0.5, culminating in a fully connected output layer. The MLP model is comprised of two fully connected layers with 100 and 10 neurons, respectively. This model is optimized with a learning rate of 0.01 and a batch size of 128.

% \vspace{-3mm}

\subsection{Classification Accuracy}

A comprehensive accuracy comparison of different aggregation methods we presented in Section~\ref{sec:approaches} is provided in Table~\ref{tab:accuracy_compare}. The experiments involve four datasets (MNIST, CIFAR-10, SVHN, and Fashion-MNIST) and the number of users set at 50 and 100. The aggregation methods compared include baseline DL without privacy-preserving aggregation, \ppdlbase, \ppdlerror, \ppdlpairs, and DLDP (Differentially Private Decentralized Learning~\cite{cheng2019towards}).

For both groups of 50 and 100 users, DL consistently achieves the highest accuracy across all datasets, serving as a benchmark for the effectiveness of decentralized learning without privacy-preserving constraints. \ppdlbase and \ppdlpairs closely match the DL's accuracy with minimal losses, indicating their capability to maintain high accuracy while offering privacy-preserving benefits. Specifically, \ppdlbase demonstrates accuracies of $96.52\%$ and $96.11\%$ for MNIST, closely following DL's $96.56\%$ and $96.24\%$ for groups of 50 and 100 users, respectively. This pattern of minimal accuracy reduction persists across other datasets, underscoring the efficiency of these methods in balancing privacy with performance.

\ppdlerror and DLDP, however, exhibit a noticeable decrease in accuracy due to their rigorous privacy constraints, as evidenced by their accuracy figures and \(\epsilon\) values. In the CIFAR-10 dataset, \ppdlerror's accuracy dips to $62.12\%$ for 50 users and $62.21\%$ for 100 users, a significant reduction from DL's $74.46\%$ and $71.24\%$. This decline in performance is echoed across the board, with \ppdlerror and DLDP consistently posting lower accuracies in exchange for stronger privacy guarantees denoted by their respective \(\epsilon\) values. For instance, in the SVHN dataset, DLDP's accuracy is $83.07\%$ for 50 users and $82.97\%$ for 100 users, compared to DL's more robust $92.92\%$ and $92.12\%$.

In conclusion, \ppdlbase and \ppdlpairs strike a good balance between privacy preservation and accuracy, closely matching the benchmarks set by the baseline DL. Meanwhile, \ppdlerror offers strong privacy guarantees by combining differential privacy  and masking techniques for DL, achieving a level of accuracy comparable to that of DLDP.

\begin{figure*}[tb!]
    \centering
    % First figure
    \begin{subfigure}{0.25\linewidth}
        \centering
        \includegraphics[width=\linewidth]{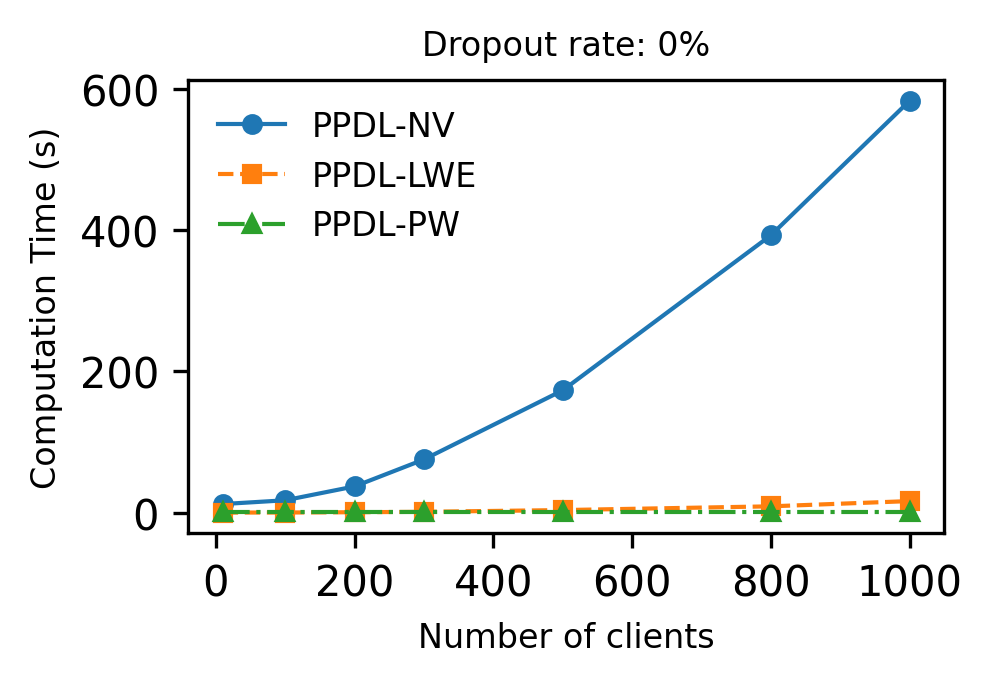}
        % \caption{First caption} % Individual caption for the first figure
    \end{subfigure}\hfill % Adjust spacing as needed
    % Second figure
    \begin{subfigure}{0.25\linewidth}
        \centering
        \includegraphics[width=\linewidth]{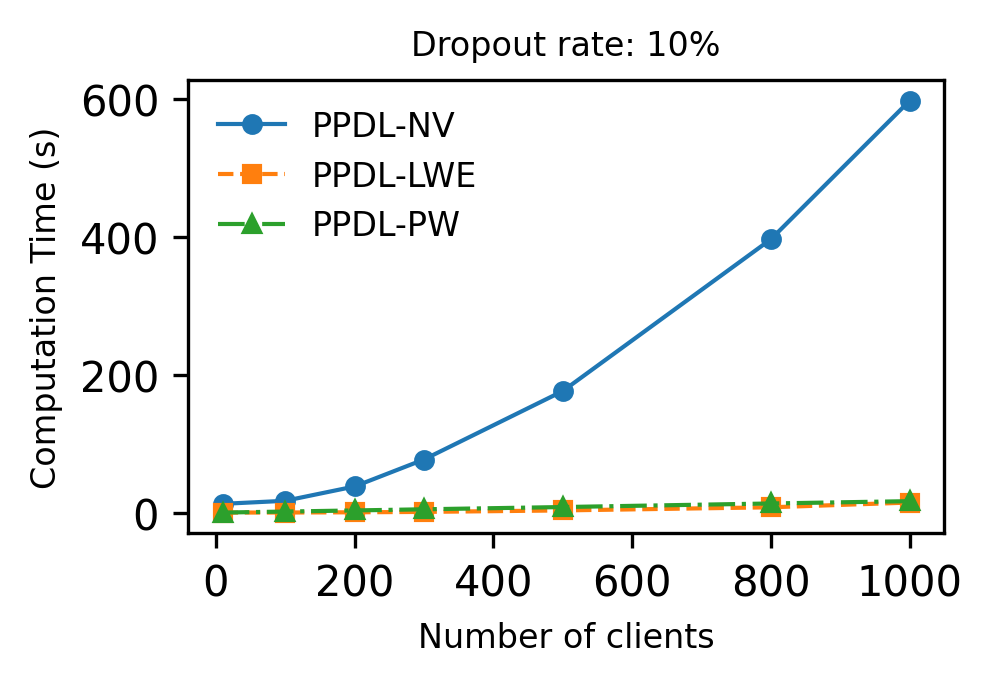}
        % \caption{Second caption} % Individual caption for the second figure
    \end{subfigure}\hfill % Adjust spacing as needed
    % Third figure
    \begin{subfigure}{0.25\linewidth}
        \centering
        \includegraphics[width=\linewidth]{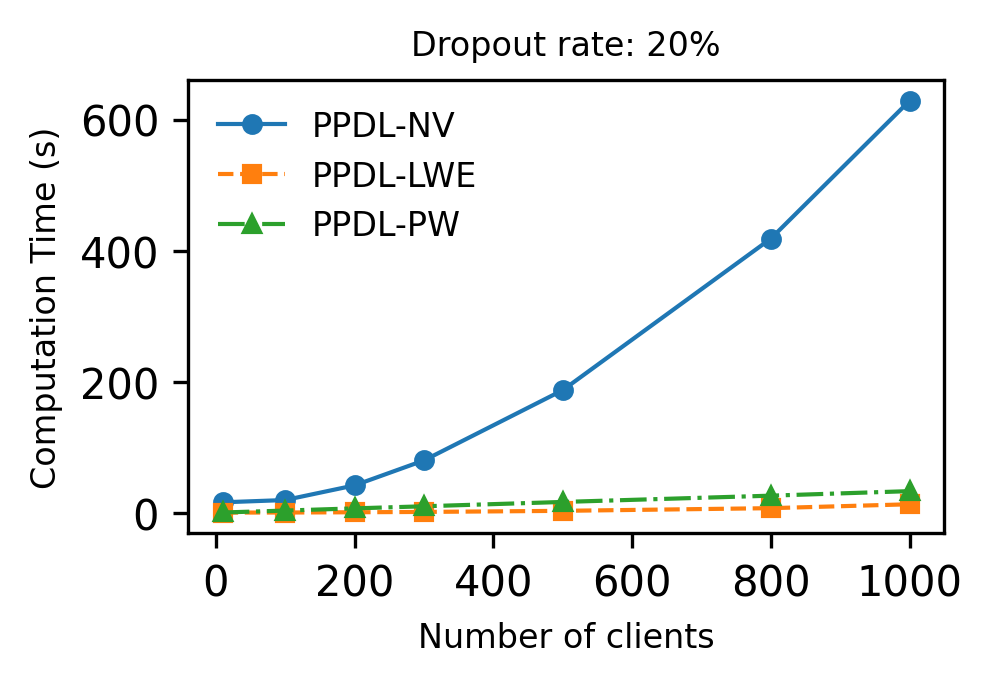}
        % \caption{Third caption} % Individual caption for the third figure
    \end{subfigure}\hfill % Adjust spacing as needed
    % Fourth figure
    \begin{subfigure}{0.25\linewidth}
        \centering
        \includegraphics[width=\linewidth]{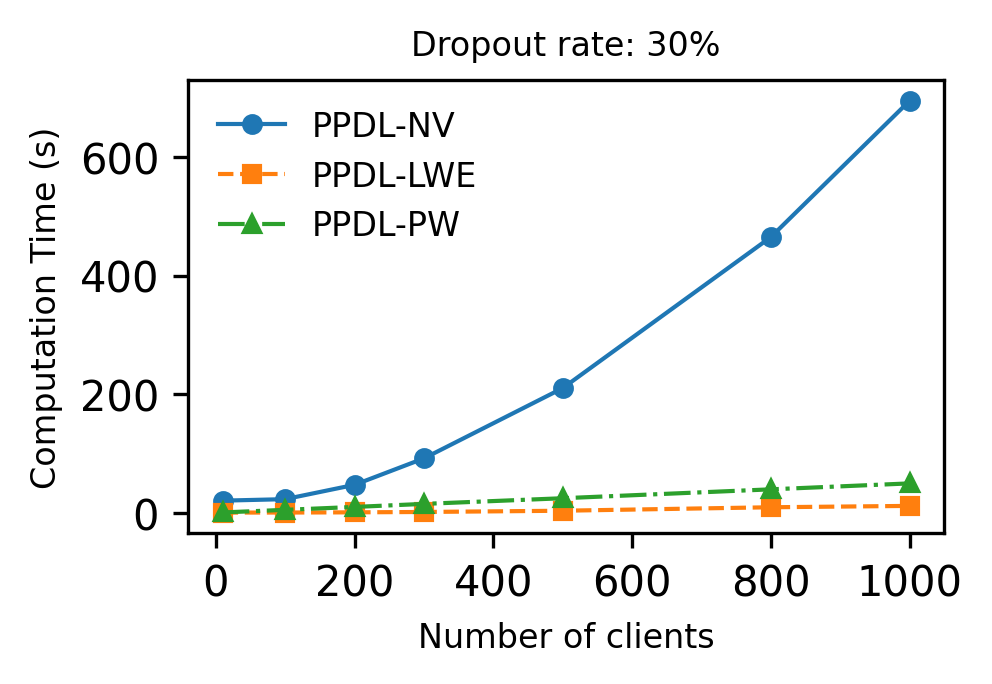}
        % \caption{Fourth caption} % Individual caption for the fourth figure
    \end{subfigure}
    
    \caption{Implications of increasing the number of clients on the efficiency of protocols.} % Main caption for all figures
    \label{fig:fractionbyz}
\end{figure*}

\begin{figure*}[tb!]
    \centering
    % First figure
    \begin{subfigure}{0.25\linewidth}
        \centering
        \includegraphics[width=\linewidth]{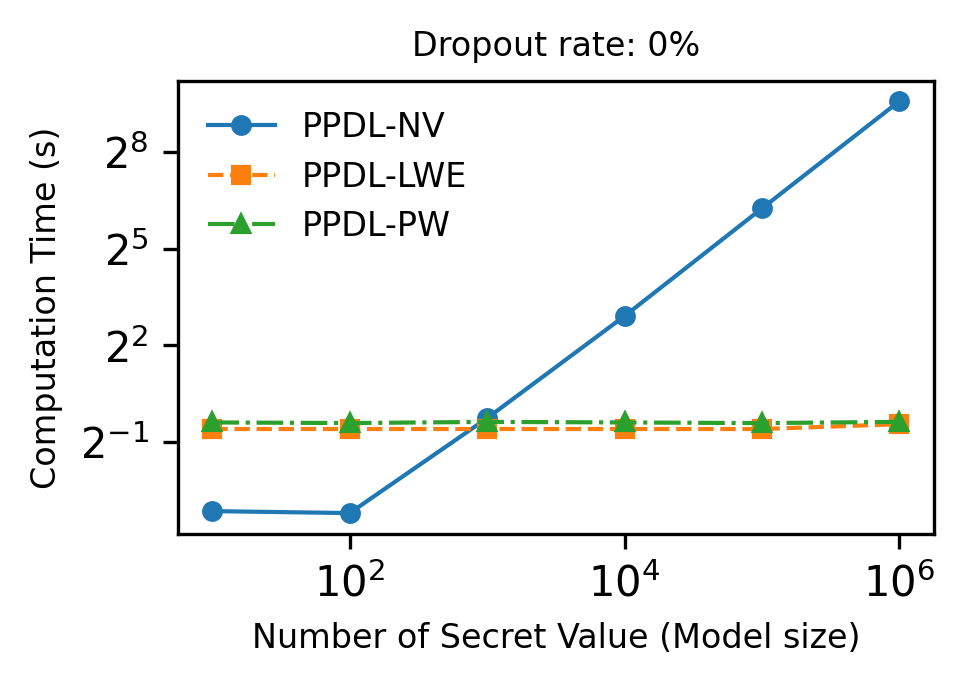}
        % \caption{First caption} % Individual caption for the first figure
    \end{subfigure}\hfill % Adjust spacing as needed
    % Second figure
    \begin{subfigure}{0.25\linewidth}
        \centering
        \includegraphics[width=\linewidth]{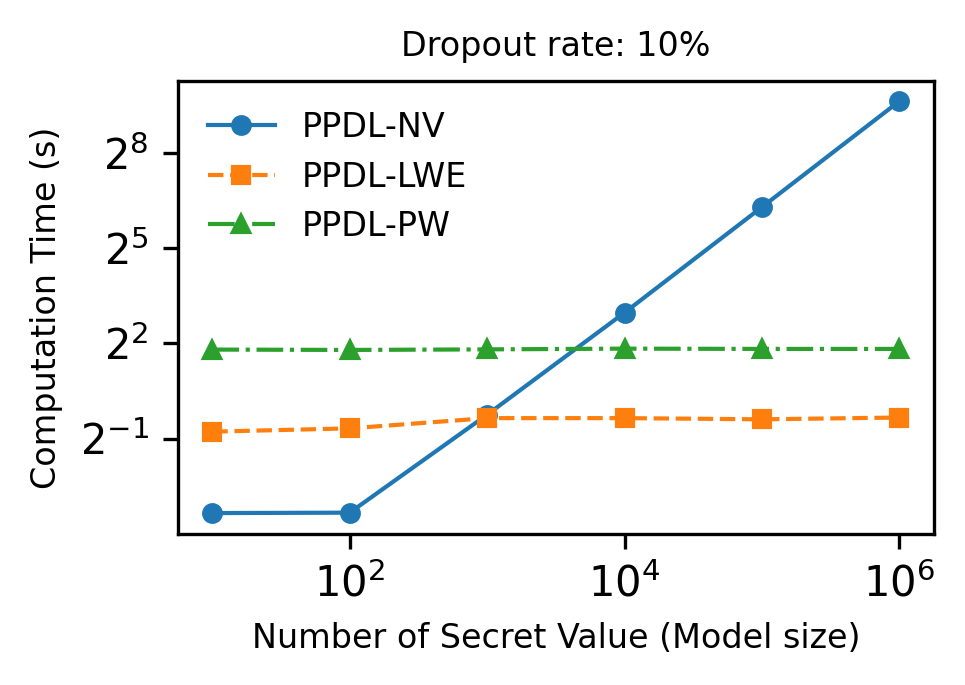}
        % \caption{Second caption} % Individual caption for the second figure
    \end{subfigure}\hfill % Adjust spacing as needed
    % Third figure
    \begin{subfigure}{0.25\linewidth}
        \centering
        \includegraphics[width=\linewidth]{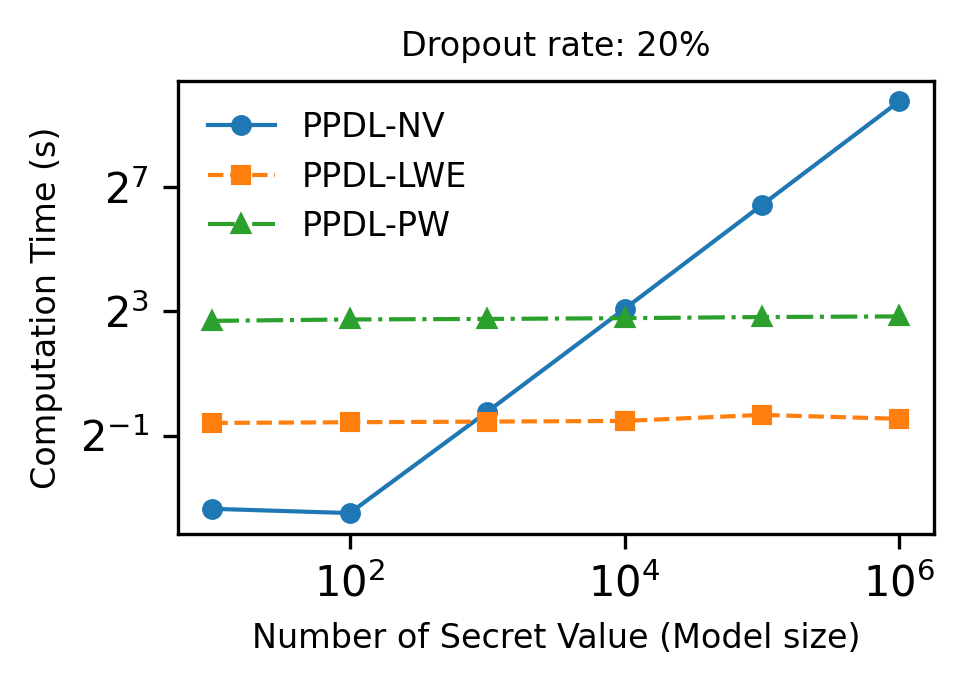}
        % \caption{Third caption} % Individual caption for the third figure
    \end{subfigure}\hfill % Adjust spacing as needed
    % Fourth figure
    \begin{subfigure}{0.25\linewidth}
        \centering
        \includegraphics[width=\linewidth]{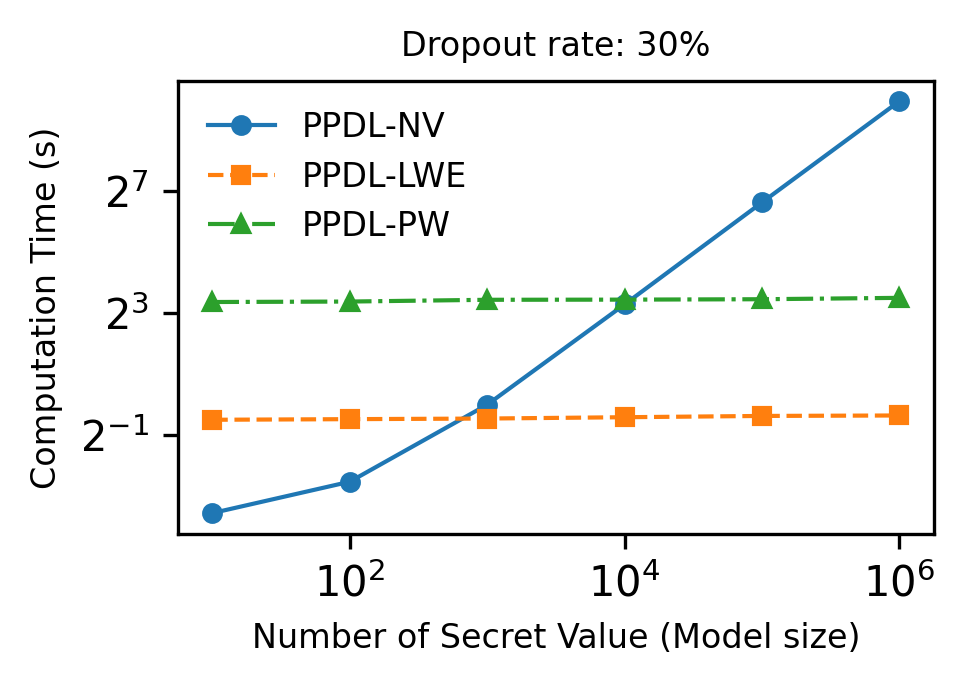}
        % \caption{Fourth caption} % Individual caption for the fourth figure
    \end{subfigure}
    
    \caption{Implications of increasing the number of model parameters on the efficiency of protocols.} % Main caption for all figures
    \label{fig:fractionmodel}
\end{figure*}

\subsection{Impact of Client Number and Dropout on Computational Overhead}

Figure~\ref{fig:fractionbyz} presents a comparative analysis of the three privacy-preserving aggregation methods for DL 
% (\ppdlbase, \ppdlerror, and \ppdlpairs)  
in terms of computation time when the number of clients and dropout ratio change. 
% It examines the effects of both the increasing number of clients and the variation in dropout rates on computation times. 
This examination is based on a single iteration of the learning process, with the model parameter count set at 50,000. Additionally, we note that 20\% of the network clients are considered semi-honest, adding another layer of complexity to the privacy-preserving mechanisms employed. 

As the number of clients increases from 10 to 1000, we observe a notable increase in computation times for all evaluated methods. \ppdlbase, despite leveraging packed Shamir secret sharing, experiences a significant rise in computation time, which varies from 12.22 seconds to 694.79 seconds. This increase underscores both the method's computational demand and its robustness in handling complex calculations. On the other hand, \ppdlerror stands out for its exceptional efficiency, with computation times only ranging from 0.16 to 16.43 seconds as the number of clients grows. This efficiency highlights \ppdlerror's potential for scalability in large-scale applications where computational resources are a critical consideration."

% As the number of clients increases from 10 to 1000, a significant escalation in computation times is observed across the evaluated methods. Despite employing packed Shamir secret sharing, \ppdlbase has a substantial  computation time, ranging from 12.22 seconds to 694.79 seconds, indicating its robustness and computational intensity. Conversely, \ppdlerror demonstrates exceptional efficiency, with computation times expanding only from 0.16 to 16.43 seconds across the client range. This suggests its suitability for large-scale applications where computational resources may be limited.

\ppdlpairs stands out for its unique response to network conditions. At a 0\% dropout rate, it maintains constant computation times of 0.71 seconds irrespective of the number of clients, showcasing unparalleled efficiency in stable conditions. However, its computation times increase linearly with the number of clients and dropout rates as network instability grows, peaking at 49.50 seconds for 1000 clients at a 30\% dropout rate. This behavior points to a direct sensitivity to dropout rates, a characteristic not as pronounced in the other methods, where increases in computation time remain relatively unaffected by variations in dropout.

\vspace{-2mm} 

\subsection{Impact of Model Size on Computational Efficiency}

The scalability of \ppdlbase, \ppdlerror, and \ppdlpairs in response to increasing model sizes and varying dropout rates is illuminated through a closer examination of their computational performance, as depicted in Figure~\ref{fig:fractionmodel}. \ppdlbase’s computational times escalate significantly with model size, increasing from 0.11 seconds for 10 secret values to 760.58 seconds for 1,000,000 secret values at a 0\% dropout rate. This trend intensifies under a 30\% dropout rate, where computational times further surge to 983.88 seconds for 1,000,000 secret values, indicating a profound sensitivity to both increased model size and higher dropout rates. In contrast, \ppdlerror exhibits extraordinary resilience to these factors, with computation times only slightly increasing from 0.66 seconds for 10 secret values to 0.73 seconds for 1,000,000 secret values, even at a 30\% dropout rate. This negligible variation suggests unparalleled stability across model size and network reliability dimensions.

\ppdlpairs, on the other hand, while demonstrating a modest increase in computation time with model size, reveals a more nuanced sensitivity to dropout rates. For example, at a 0\% dropout rate, computation times increase slightly from 0.76 seconds for 10 secret values to 0.77 seconds for 1,000,000 secret values. However, as the dropout rate reaches 30\%, computation times for \ppdlpairs rise more noticeably, from 0.76 seconds for 10 secret values to 11.25 seconds for 1,000,000 secret values. This pattern suggests that while \ppdlpairs is somewhat resilient to model size increases, its efficiency is more significantly impacted by higher dropout rates, especially as the model size grows.

These numerical comparisons reveal the critical interplay between model size and dropout rates in decentralized learning environments. \ppdlbase, though scalable to a degree, faces substantial challenges in maintaining computational efficiency with larger model sizes and higher dropout rates. \ppdlerror stands out for its robustness, showing minimal impact from either increasing model sizes or varying dropout rates, making it a highly reliable choice for large-scale applications. \ppdlpairs occupies a middle ground, offering reasonable scalability and efficiency but with a noticeable decrease in performance under higher dropout rates and larger model sizes.

\begin{table}[t]
\resizebox{\linewidth}{!}{
\begin{tabular}{c|c|c|c|c|c|}
\cline{2-6}
                              & \ppdlbase & \ppdlerror  & \ppdlpairs & $\mathrm{D}^2$-MHE & Paillier-HE \\ \hline
\multicolumn{1}{|c|}{MNIST}   & 20.93  & 0.24 & 0.21    & 26.2~\cite{xu2023secure}               & 59.54~\cite{xu2023secure}       \\ \hline
\multicolumn{1}{|c|}{CIFAR10} & 513.18 & 0.24 & 0.21    & 62.19~\cite{xu2023secure}              & 106.42~\cite{xu2023secure}      \\ \hline
\end{tabular}
}
\caption{Computational overhead of one iteration of various privacy-preserving approaches (Seconds)}
\label{tab:he}
\end{table}

% \vspace{-10mm}
\subsection{Computation overhead comparison}

In addition to evaluating our proposed approaches, we also compare our protocols with two homomorphic encryption-based DL aggregation protocols, namely $\mathrm{D}^2$-MHE~\cite{xu2023secure} and Paillier-HE~\cite{xu2023secure}, to provide a comprehensive analysis of computational overhead across different privacy-preserving mechanisms when applied to the MNIST and CIFAR-10 datasets. Table~\ref{tab:he} illustrates the comparative computational demand for one iteration among various approaches, including \ppdlbase, \ppdlerror, \ppdlpairs, along with the aforementioned homomorphic encryption-based schemes.

The $\mathrm{D}^2$-MHE scheme capitalizes on recent advancements in homomorphic encryption, explicitly employing the Brakerski-Fan-Vercauteren (BFV) framework, aiming to minimize computational costs while securely updating gradients in a decentralized training context. On the other hand, Paillier-HE represents a classical approach to homomorphic encryption, adapted here for decentralized settings. The performance data for both $\mathrm{D}^2$-MHE and Paillier-HE were adopted from the study \cite{xu2023secure}, eschewing direct implementation for these comparisons. We maintained identical model architectures and configurations across all evaluated protocols to ensure a fair and accurate comparison.

For the MNIST dataset, \ppdlbase's method exhibits a relatively high computational overhead at 20.93 seconds, which dramatically increases to 513.18 seconds for the more complex CIFAR-10 dataset, indicating its less favorable scaling with data complexity. In stark contrast, \ppdlerror and \ppdlpairs methods demonstrate exceptional efficiency, with negligible overheads of 0.24 and 0.21 seconds, respectively, for both datasets. This stark efficiency highlights their suitability for decentralized scenarios where minimizing computational burden is crucial. However, the $\mathrm{D}^2$-MHE and Paillier-HE approaches present significantly higher computational times. For MNIST, $\mathrm{D}^2$-MHE stands at 26.2 seconds, escalating to 62.19 seconds for CIFAR-10, while Paillier-HE jumps from 59.54 to 106.42 seconds, respectively. These results reflect the inherent computational overhead associated with HE techniques, particularly Paillier-HE's intensive modular exponential operations and the collaborative decryption process among multiple users, which, despite its privacy advantages, results in considerable computational overhead.

\subsection{Computational and communication complexity}

In this section, we derive the computational and communication complexities associated with each of the proposed \ppdl protocols. 
% Table~\ref{tab:complexity} provides a summary of relevant complexities and protocols.

% \begin{table}[t]
% \resizebox{\linewidth}{!}{
% \begin{tabular}{|c|c|c|c|}
% \hline
% Complexity & \ppdlbase & \ppdlerror  & \ppdlpairs  \\ \hline
% Computational    & \(O(nT^2/K)\)  & \(O(nmT/K + nT^2/K)\) & \(O(nmT/K)\), \(O(nT^2/K)\)\\ 
% Communication   & \(O(n^2)\)  & \(O(k^2m + k^2n)\) & \(O(n^2 + mn)\) \\ \hline
% \end{tabular}
% }
% \caption{Summary of computational and communication complexity of proposed \ppdl. The two complexities of \ppdlpairs are for share generation and reconstruction respectively.}
% \label{tab:complexity}
% \end{table}

\subsubsection{\ppdlbase}

\textbf{Computational complexity}: \ppdlbase protocol, employing Packed Shamir's Secret Sharing, has a computational complexity that unfolds across a few key steps. Initially, clients generate packed shares of their local model update \(w_i\) with a complexity of \(O(nT/K)\), leveraging the efficiency of packing multiple elements per share for \(n\) clients, \(T\) as the threshold, and \(K\) as the packing factor. The aggregation of shares from \(n-1\) peers follows a linear complexity of \(O(n)\), thanks to the simple addition of packed shares. The final, most complex phase is the reconstruction of the aggregated value from shares, at \(O(nT^2/K)\), due to polynomial interpolation for unpacking. The protocol's overall complexity is thus primarily determined by this reconstruction phase, leading to an \(O(nT^2/K)\) complexity.

% \textbf{Computational complexity}: In this protocol that utilizes Packed Shamir's Secret Sharing, the computational complexity encompasses several sequential steps: Initially, each client generates packed shares of their local model update \(w_i\) with complexity \(O(nT/K)\), reflecting the efficiency of packing multiple elements into each share for \(n\) clients, with \(T\) as the threshold and \(K\) as the number of elements packed. Following this, clients aggregate the received shares from \(n-1\) peers, a step characterized by its linear complexity \(O(n)\) due to the straightforward addition of these packed shares. The final step involves reconstructing the aggregated value from the collected shares, which bears a complexity of \(O(nT^2/K)\), accounting for the polynomial interpolation required to unpack and decode the aggregated model from the packed shares. Considering all steps, the overall computational complexity of the protocol is dominated by the reconstruction phase, leading to an overall complexity of \(O(nT^2/K)\).

\textbf{Communication complexity:} This protocol's communication complexity arises from two key operations: share distribution and aggregated share broadcasting. Each client sends their packed shares to \(n-1\) others and receives an equivalent number from every other client, resulting in a communication complexity of \(O(n^2)\). This stems from the need for each of the \(n\) clients to interact with all others. The aggregation step, involving each client broadcasting their aggregated share to all, also reflects an \(O(n^2)\) complexity due to the peer-to-peer exchange among \(n\) participants. Thus, the protocol's overall communication complexity remains \(O(n^2)\), underscored by the extensive interactions necessary for decentralized aggregation.

\subsubsection{\ppdlerror}

\textbf{Computational complexity:} The computational complexity for each client is primarily shaped by generating random vectors (\(O(m)\)), matrix-vector multiplication (\(O(mn)\)), and leveraging packed Shamir's Secret Sharing (\(O(nmT/K)\) for sharing and \(O(nT^2/K)\) for reconstruction). Given \(m\) exceeds the length of the secret vector \(s\), the initial vector generation simplifies to \(O(m)\). The subsequent operations, particularly the reconstruction of packed secret shares, culminate in a dominant \(O(nT^2/K)\) complexity. Thus, the overall client-side complexity is effectively summarized by \(O(nmT/K + nT^2/K)\).

\textbf{Communication complexity:} In \ppdlerror protocol, the communication complexity for each participant is derived from a series of direct, peer-to-peer exchanges. Specifically, every client sends their encrypted gradient vector of size \(m\) to \(k-1\) peers and receives an equivalent number of encrypted vectors from others, leading to a communication load of \(O(km(k-1))\). Additionally, the aggregate mask, also of size \(m\), is computed using multi-party computation and subsequently distributed among \(k-1\) peers, adding \(O(mk(k-1))\) to the communication cost. Secure aggregation, facilitated through SSS, requires each client to distribute secret shares of their vector \(s\) to all \(k-1\) other clients, contributing \(O(k^2n)\) to the overall complexity. Therefore, when accounting for both sending and receiving operations in this decentralized scenario, the communication complexity for each client effectively scales as \(O(k^2m + k^2n)\).

\subsubsection{\ppdlpairs}

\textbf{Computational complexity:} It is defined by two primary operations: \(O(n^2)\) for generating \(t\)-out-of-\(n\) Shamir secret shares for the privacy key across \(n\) participants, reflecting the quadratic nature of polynomial computations required for secure share distribution, and \(O(mn)\) for the linear complexity associated with generating pseudorandom values for each participant, where \(m\) denotes the size of the input vector. Furthermore, secret reconstruction, particularly exigent during participant dropouts, demands a computational complexity of \(O(n^3)\) due to the intricacies of polynomial interpolation. To enhance computational efficiency, our work incorporates packed Shamir's Secret Sharing, effectively reducing the complexity to \(O(nmT/K)\) for share generation and \(O(nT^2/K)\) for reconstruction, with \(T\) representing the threshold for reconstruction and \(K\) the number of secrets integrated into each polynomial. 

\textbf{Communication complexity:} Initially, each participant exchanges a single public key with every other participant, resulting in \(n(n-1)\) total public key exchanges across the network. Concurrently, secret shares of the privacy key are distributed similarly, with each participant sending and receiving a secret share to and from each of the \(n-1\) other participants, doubling the count of exchanges and reinforcing the quadratic component of the communication complexity. Additionally, each participant broadcasts a masked data vector of size \(m\) to all \(n-1\) other participants, further contributing to the communication load with \(m(n-1)\) units of data sent by each participant. The overall communication complexity per participant is \(O(n^2 + mn)\), reflecting the quadratic growth from pairwise key and secret share exchanges (\(n^2\) for large \(n\)) and the linear growth from broadcasting masked data vectors (\(mn\)).

\section{Conclusion}
This study delves into the challenges and necessities of privacy preservation and efficiency in decentralized learning environments. We introduce three novel protocols designed for secure aggregation that directly address DL's unique challenges, such as heightened risks of information leakage and data integrity issues in the face of variable client participation.
% and vulnerabilities to privacy breaches, ensuring robust privacy protection and data integrity despite variable client participation. 
Additionally, exploring mechanisms for dropout resilience strengthens the learning process against disruptions caused by client dropouts, enhancing the system's overall reliability and effectiveness. Through rigorous experimental evaluation across various datasets and scenarios, we have empirically validated the effectiveness of our proposed methods. Our experiments, testing our protocols over diverse datasets to ensure robustness and generalizability, demonstrate their computational efficiency and resilience to network fluctuations while preserving privacy. The results of the experiments reaffirm the feasibility of implementing privacy-preserving decentralized learning at scale. It provides valuable insights into optimizing aggregation protocols to balance the trade-offs between privacy, accuracy, and efficiency. This work lays the groundwork for future progress in developing more efficient decentralized learning architectures.

% \begin{acks}
% \end{acks}

\bibliographystyle{ACM-Reference-Format}
\bibliography{ref}

\end{document}